\documentclass[conference]{IEEEtran}
\normalsize

\usepackage{amsmath,amssymb,amsthm}
\usepackage{xifthen}
\usepackage{mathtools}
\usepackage{enumerate}
\usepackage{microtype}
\usepackage{xspace}
\usepackage{bm}
\usepackage{fancyhdr}
\usepackage{lastpage}
\usepackage{bbm}

\newcommand{\mbs}[1]{\pmb{#1}}
\newcommand{\vect}[1]{{\lowercase{\mbs{#1}}}}

\newcommand{\T}{{\scriptscriptstyle\mathsf{T}}}

\newcommand{\cond}{\,\vert\,}
\renewcommand{\Re}[1][]{\ifthenelse{\isempty{#1}}{\operatorname{Re}}{\operatorname{Re}\left(#1\right)}}
\renewcommand{\Im}[1][]{\ifthenelse{\isempty{#1}}{\operatorname{Im}}{\operatorname{Im}\left(#1\right)}}

\newcommand{\pv}{\vect{p}}
\newcommand{\qv}{\vect{q}}
\newcommand{\rv}{\vect{r}}
\newcommand{\sv}{\vect{s}}

\newcommand{\uv}{\vect{u}}
\newcommand{\vv}{\vect{v}}

\newcommand{\xv}{\vect{x}}
\newcommand{\yv}{\vect{y}}

\newcommand{\zerov}{\vect{0}}
\newcommand{\onev}{\vect{1}}

\newcommand{\alphav}{\vect{\alpha}}
\newcommand{\betav}{\vect{\beta}}

\newcommand{\Ic}{{\mathcal I}}
\newcommand{\Jc}{{\mathcal J}}
\newcommand{\Kc}{{\mathcal K}}
\newcommand{\Lc}{{\mathcal L}}
\newcommand{\Mc}{{\mathcal M}}

\newcommand{\Qc}{{\mathcal Q}}
\newcommand{\Rc}{{\mathcal R}}
\newcommand{\Sc}{{\mathcal S}}

\newcommand{\Uc}{{\mathcal U}}

\newcommand{\Xc}{{\mathcal X}}
\newcommand{\Yc}{{\mathcal Y}}

\newcommand{\CN}[1][]{\ifthenelse{\isempty{#1}}{\mathcal{N}_{\mathbb{C}}}{\mathcal{N}_{\mathbb{C}}\left(#1\right)}}

\renewcommand{\P}[1][]{\ifthenelse{\isempty{#1}}{\mathbb{P}}{\mathbb{P}\left(#1\right)}}
\newcommand{\E}[1][]{\ifthenelse{\isempty{#1}}{\mathbb{E}}{\mathbb{E}\left[#1\right]}}
\newcommand{\I}[1][]{\ifthenelse{\isempty{#1}}{\mathbb{I}}{\mathbb{I}\left\{#1\right\}}}
\renewcommand{\det}[1][]{\ifthenelse{\isempty{#1}}{\mathrm{det}}{\text{det}\left(#1\right)}}
\newcommand{\trace}[1][]{\ifthenelse{\isempty{#1}}{\mathrm{tr}}{\text{tr}\left(#1\right)}}
\newcommand{\rank}[1][]{\ifthenelse{\isempty{#1}}{\mathrm{rank}}{\text{rank}\left(#1\right)}}
\newcommand{\diag}[1][]{\ifthenelse{\isempty{#1}}{\mathrm{diag}}{\text{diag}\left(#1\right)}}
\newcommand{\Cov}[1][]{\ifthenelse{\isempty{#1}}{\mathsf{Cov}}{\mathsf{Cov}\left(#1\right)}}

\newtheorem{proposition}{Proposition}

\newtheorem{theorem}{Theorem}

\renewcommand{\rv}[1]{{\mathrm{#1}}}

\newcounter{enumi_saved}
\usepackage{answers}
\Newassociation{solution}{Solution}{solutionfile}

\AtBeginDocument{\Opensolutionfile{solutionfile}[\jobname]}
\AtEndDocument{\Closesolutionfile{solutionfile}\clearpage
}

\IfFileExists{MinionPro.sty}{
}{
}

\newcommand{\overbar}[1]{\mkern 1.5mu\overline{\mkern-1.5mu#1\mkern-1.5mu}\mkern 1.5mu}
\usepackage{stfloats}
\usepackage{booktabs}
\allowdisplaybreaks[4]
\usepackage{cite}

\begin{document}
	\title{Exploiting Location Information to Enhance Throughput in Downlink V2I Systems}
	\author{\IEEEauthorblockN{Zheng Li$\stackrel{\ddag}{}$$\stackrel{^\S}{}$,  Sheng Yang$\stackrel{\ddag}{}$, and Thierry Clessienne$\stackrel{^\S}{}$}
		\IEEEauthorblockA{$^\ddag$ L2S, CentraleSup\'elec-CNRS-Universit\'e Paris-Sud, 91192, Gif-sur-Yvette, France\\
			$^\S$Orange Labs Networks, 92326, Ch\^atillon, France\\
			Email:\{zheng.li, sheng.yang\}@centralesupelec.fr, thierry.clessienne@orange.com}}
	\maketitle
	\begin{abstract}
		Vehicle-to-Infrastructure (V2I) technology, combined with millimeter wave (mmW) networks, may support high data rates for vehicular communication and therefore provides a whole new set of services. 
		However, in dense urban environment, pedestrians or buildings cause strong blockage to the narrow beams at mmW,
		severely deteriorating the transmission rate. In this work, we model the downlink mmW V2I system as a simple
		erasure broadcast channel where the erasure~(blockage) event is considered as the state of the channel. While state
		feedback can be obtained through protocols such as Automatic Repeat reQuest (ARQ), we also assume that the current state can be estimated
		through the location information shared by the communication peers. We evaluate, through an information-theoretic
		approach, the achievable downlink rate in such a system. Despite its highly theoretical nature, our result sheds light on
		how much the location information can contribute to improve the downlink date rate, e.g., as a function of the
		mobility~(velocity) of the vehicles.

		
	\end{abstract}
	\IEEEpeerreviewmaketitle
	
	\section{Introduction}\label{sec:introduction}
	
	The Vehicle-to-Infrastructure (V2I) technology, as an important part of the new-generation vehicular communication
	system, aims to offer high data rate that can fuel a wide range of services. As the conventional 
	wireless microwave bands are becoming increasingly crowded, the millimeter wave~(mmW) spectrum, containing
	a massive amount of raw bandwidth, came to be a promising candidate for the V2I
	systems~\cite{kong2017millimeter,choi2016millimeter}. 
	In addition to accessing larger bandwidths, mmW can also allow very compact antenna arrays to provide high directional beamforming and thus interference reduction due to narrow beams.
	Delivering advanced services to vehicles requires the deployment of mmW microcellular radio sites in the vehicular urban areas.
	
	Compared to the channels at microwave frequencies, however, mmW
	channels are more sensitive, due to the smaller wavelengths, to blockage losses, especially in urban streets where signals are blocked by high buildings, vehicles or pedestrians. 
	Even for high mmW microcellular radio site density, the blockage cannot be completely excluded or predicted
	as the beams are narrow and the vehicles are usually moving. It is reasonable to model the blockage as an erasure
	event~(package erasure or symbol erasure depending on the physical layer) from a higher level. Indeed, such an
	approach enables us to somehow ignore~(and thus be less dependent on) the physical layer aspects and to focus on the
	important system level parameters in order to obtain useful insights.  
	
	\begin{figure}
		\centering
		\includegraphics[width=0.9\columnwidth]{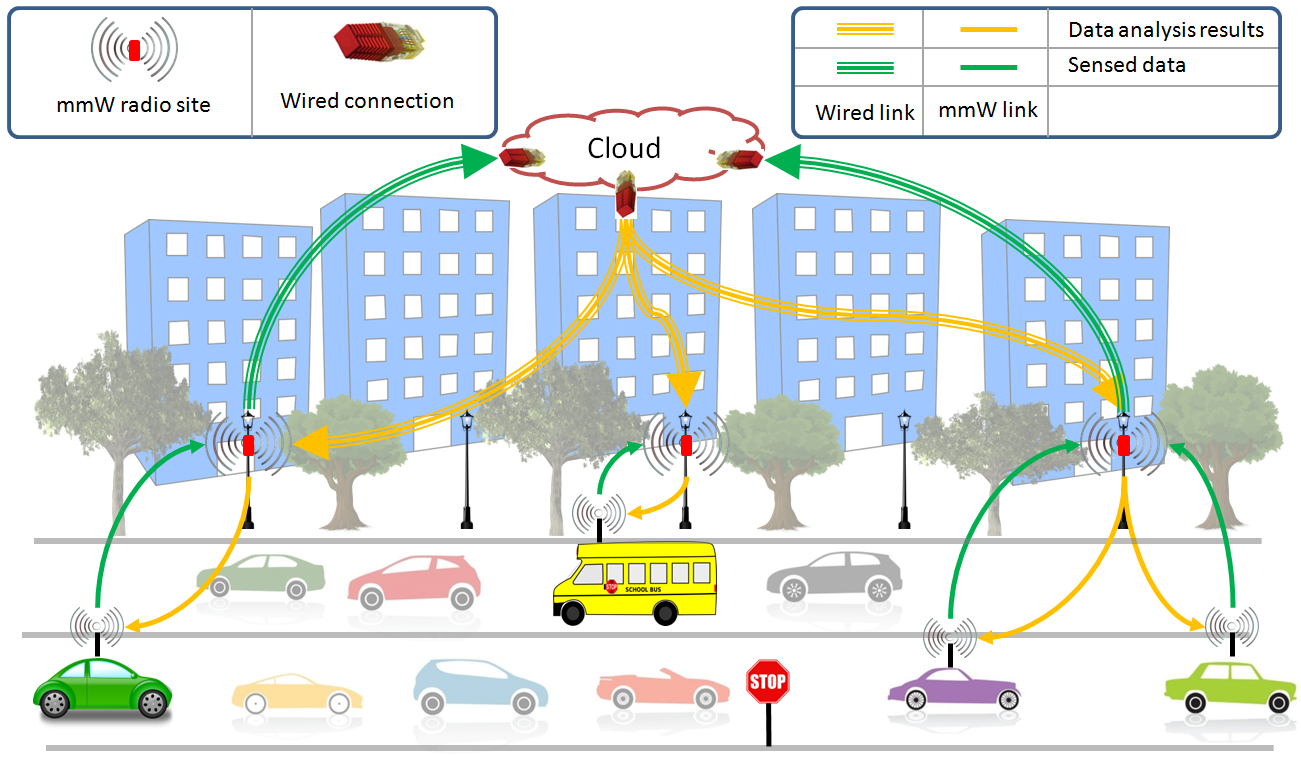}
		\caption{The mmW V2I system under consideration.}\label{fig:background}
	\end{figure}
	
	As an example, let us consider a mmW V2I system as illustrated in Fig.\ref{fig:background}. Here the cloud 
	gathers and processes the information inside~(e.g. traffic data, object recognition tasks) or outside~(e.g. regular
	internet data, multimedia contents) the network, then communicates to $K$ vehicles with their own desired
	messages. 
	All the talking vehicles are connected to the cloud through mmW wireless radio sites, which are themselves connected
	to the cloud (the wireless radio sites operating as relays). Here the wired cloud-radio site links are supposed to be
	perfect (without loss of information) while the wireless vehicle-radio site links suffer from blockage. As suggested
	in the previous paragraph, we consider the downlink channel from the cloud to the vehicles as a
	1-to-$K$ Erasure Broadcast Channel (EBC) despite the presence of the wireless radio sites. Then,
	whether a packet or symbol can be successfully decoded by a particular vehicle depends strongly on whether the vehicle is
	in blockage. We refer to the latter as the state of the channel. 
	
	It turned out that if the state is known to the cloud, even with a long delay such that the state is completely
	outdated, such information can still increase substantially the channel
	capacity~\cite{wang2012capacity,gatzianas2013multiuser}, even if this outdated state information is noisy \cite{shayevitz2013capacity,venkataramanan2013achievable} or rate-limited \cite{wu2016coding}. In fact, such binary state information is usually fed back
	to the transmitter, i.e., the cloud, with an ACK/NACK mechanism in practical communication systems, which makes the setup quite
	realistic. In this work, we assume that the cloud transmitter can somehow estimate~(imperfectly) the current state, in
	addition to the perfect information on the past channel states obtained with feedback. Indeed, such an
	assumption can be justified by the fact that the cloud has centralized information on the vehicles that can be
	exploited to predict whether or not the channel is in blockage. One such example is the location information.
	Provided that the location of the vehicles is tracked in real time by the cloud, the channel condition can be
	predicted based on factors such as the propagation environment and the velocity of the vehicles. Hence, our main
	interest in this work is to evaluate the potential throughput gain brought by such additional information.  
	
	The main technical contributions of our work are as follows. First, we propose a new scheme that exploits
	both the current and past state information, and derive the corresponding rate region for  $K$-receiver EBC.
	Second, we show that in the two-receiver case, the proposed scheme is indeed optimal. As compared to a previously
	proposed scheme in \cite{li2017capacity}, the new ingredient is the mixture of private information in addition to the
	separate transmission. Third, we exploit the general result to evaluate the potential gain from the location
	information for a mmW V2I network. Although the theoretical result is well beyond the V2I scenario, we do believe
	that this scenario is one of the few cases for which the underlying assumptions of the general result can be
	realistic. 
	In contrast to previous works on similar setups \cite{heindlmaier2018capacity, pantelidou2009cross} which use queuing-theoretic tools,
	we adopt an Information-Theoretic~(IT) approach to derive the achievable rate region. Our scheme is based
	on standard IT tools such as random coding arguments, block Markov coding, joint source-channel coding and
	typicality decoding. 

	The paper is organized as follows. Section~\ref{sec:system_model}  describes the system model while Section~\ref{sec:achieve} presents the optimal rate region for a two-receiver mmW V2I system and discusses the potential gain from additional location information. The novel scheme for $K$-receiver EBC is explained in Section~\ref{sec:proof}. Finally, Section~\ref{sec:conclusion} concludes the paper. Some details are relegated to appendix. 
	
	\emph{Notation}: Throughout the paper, 
	vectors follow the column convention and are in bold letters, e.g., 
	the vectors of ones and zeros are denoted by
	$\onev$ and $\zerov$, respectively. 
	$\Kc := \{1,\ldots,K\}$ is the universe vehicle (receiver) set, calligraphic capitalized letters $\Jc,\Ic,\Lc,\Uc$ represent some subsets of $\Kc$,  we always assume that $|\Jc|,|\Ic|,|\Lc|>1$ (not $|\Uc|$) in this paper. We use $\uv\preceq\vv$ to mean that $u_i\le v_i$,
	$\forall\,i$. 
	
	\section{System Model}\label{sec:system_model}
	
	Let us consider the downlink communication of a mmW V2I system, with one cloud transmitter and $K$ vehicle receivers. 
	We assume that at each time slot $t$, the transmitter sends a signal $x_t\in\Xc$ where $\Xc$ is an arbitrary
	alphabet. For instance, if $x_t$ is a packet of $l$ bits, then $\Xc = \mathbb{F}_2^l$. Each receiver $k$ recovers
	$y_{k,t}$ that can be either exactly $x_{t}$ or ``$?$''~(erased). The system is therefore equivalent to a 1-to-$K$ EBC. For convenience, we introduce the state variable
	$S_{k,t} = 0$ for erasure and $S_{k,t} = 1$ otherwise. Let $\rv{S}_t := \bigl[ S_{1,t}, \ldots, S_{K,t} \bigr]$
	denote the global state at time $t$. We assume that the global state is available to the transmitter with one-slot delay, i.e., the transmitter knows perfectly $\rv{S}_{t-1}$ from time slot $t$ on via some feedback mechanism such as Automatic Repeat reQuest (ARQ). In addition, at time $t$, the transmitter can obtain an estimate of current global state
	$\hat{\rv{S}}_t:= \bigl[ \hat{S}_{1,t}, \ldots, \hat{S}_{K,t} \bigr]$ through side information such as the GPS location information collected regularly from the
	vehicles. For tractability, we make the following assumptions. 
	First, the joint process $\{\rv{S}_t, \hat{\rv{S}}_t\}_t$ is stationary in time, with
	joint distribution $\P(\rv{S}_t=s, \hat{\rv{S}}_t=\hat{s})$ for $s,\hat{s}\in\{0,1\}^K$ that does
	not depend on $t$. Define the probability vector $\pv_{s}$ for a given state $s\in \{0,1\}^K$ as
	\vspace{-0.1cm}
	\begin{align}
		\pv_{s} &:= \bigl[ \P(\rv{S}=s, \hat{\rv{S}}=\hat{s}):\quad \hat{s}\in \{0,1\}^K \bigr],
		\label{eq:prob_vec}  
	\end{align}%
	
	\vspace{-0.1cm}
	\noindent such that the marginal distribution is $p_{s} = \pmb{1}^\T \pv_{s}$. We also define
	$\pv_{\bar{s}} := \bigl[ \P(\rv{S}\ne s, \hat{\rv{S}}=\hat{s}):\quad
	\hat{s}\in \{0,1\}^K \bigr]$.   
	Then, we suppose that the following Markov chain holds
	\vspace{-0.1cm}
	\begin{align}
		\cdots \leftrightarrow \rv{S}_{t-1} \leftrightarrow \hat{\rv{S}}_{t}
		\leftrightarrow \rv{S}_{t} \leftrightarrow \hat{\rv{S}}_{t+1} \leftrightarrow \cdots
	\end{align}%
	
	\vspace{-0.1cm}
	\noindent
	In other words, the predictor exploits all the available information in an optimal way to obtain
	$\hat{\rv{S}}_t$, such that given $\hat{\rv{S}}_t$, the original information is irrelevant to estimating
	$\rv{S}_t$.   
	Finally, we also assume that the marginal distributions of the processes $\{\rv{S}_t\}_t$ and
	$\{\hat{\rv{S}}_t\}_t$ are identical. As we shall show later, the performance of the system
	in this abstract model only depends on the set of probability vectors defined in
	\eqref{eq:prob_vec}.  
	To make the model more concrete, let us consider the following toy example.  
	
	\subsection*{Toy Model}
	
	
	Assume that the mmW radio sites are spatially distributed in $\mathbb{R}^2$ as a homogeneous Poisson Point
	Process~(PPP) of density $\lambda$. To model the blockage in mmW, we adopt
	the LoS ball model proposed in \cite{bai2015coverage}, i.e., a vehicle receives the signal
	perfectly within the distance $R_B$ to any wireless radio sites, and is in blockage (signal completely erased) otherwise.
	Thus, we have
	\vspace{-0.1cm}
	\begin{align}
		\mathbb{P}({S}_{k,t}=0)=\mathbb{P}(\hat{S}_{k,t}=0)=e^{-\lambda A},\label{eq:def_S}
	\end{align} 
	
	\vspace{-0.1cm}
	\noindent where $A:=\pi R_B^2$ is the area of a circle of radius $R_B$. If we assume by simplicity that
	the current state for vehicle~$k$ is uniquely predicted with the GPS location information sent
	by the vehicle, then the actual state may differ from the estimated one depending on the
	velocity of the vehicle. Indeed, with low mobility, the estimated state should be rather
	accurate, whereas with high mobility, the location information becomes completely outdated.
	Let $T_s$ be the delay between the acquisition time of the GPS information and time slot
	$t$, and $v_k$ be the velocity of vehicle $k$. Then for example, $\P(S_{k,t}\!=0|\hat{S}_{k,t}\!=0)$ is the probability that there is no mmW radio sites in the shadowed area shown in Fig. \ref{fig:A_k}. For each vehicle $k$, the distance between the real location and the outdated GPS location is $v_kT_s$, the shadowed area can then be written as $A_k:=\bigl(A+\frac{v_kT_s}{2}\sqrt{4R_B^2-(v_kT_s)^2}-2R_B^2\arccos\big(\frac{v_kT_s}{2R_B}\big)\bigr)^+$,
	which increases with the velocity $v_k$ from $0$ to $A$, and stays at $A$ when $v_kT_s\ge 2R_B$.
	\begin{figure}
		\centering
		\includegraphics[width=0.6\columnwidth]{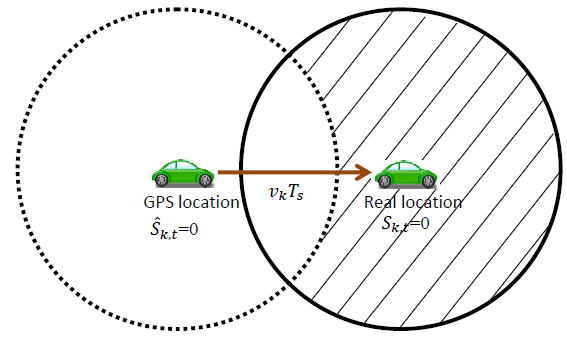}
		\caption{Illustration of the probability $\P(S_{k,t}=0|\hat{S}_{k,t}=0)$, the shadowed area is denoted as $A_k$.}\label{fig:A_k}
		\vspace{-0.1cm}
	\end{figure}
	Following the property of PPP and LoS ball model, we obtain
	\begin{align}
		&\mathbb{P}(S_{k,t}\!=0|\hat{S}_{k,t}\!=0)=e^{-\lambda A_k},\label{eq:def_cond_S1}\\
		&\mathbb{P}(S_{k,t}=0|\hat{S}_{k,t}=1)=\frac{e^{-\lambda A}(1-e^{-\lambda A_k})}{1-e^{-\lambda A}}.\label{eq:def_cond_S2}
	\end{align}
	
	\vspace{-0.1cm}
	\noindent
	Note that the distribution does not depend on the time
	index, which is in accordance with the stationarity of our model. A further simplification is
	to assume that the vehicles are spatially independent, i.e., 
	$\P(\rv{S}\!=\!s, \hat{\rv{S}} \!=\! \hat{s})=\prod\nolimits_{k} {\P(\hat{S}_k =
		\!\hat{s}_k\!)}{\P(S_k\!=\!s_k\cond\hat{S}_k = \!\hat{s}_k\!)},\forall
	s,\hat{s}\in\{0,1\}^{K}$. 
	
	The above toy example is interesting since it boils down the complex V2I system into three
	important features: density of the mmW radio sites deployment $\lambda$, mobility of the vehicles $\{v_k\}$,
	and the timeliness of the location information. 
	While we emphasize that our result is not restricted to the example, we shall use it repeatedly for illustration purpose and to provide useful insights.
	
	\vspace{-0.1cm}
	\section{Two-receiver V2I system}\label{sec:achieve}
	
	We start with the simplest setting in which there are only two vehicles in the system. In this
	section, we first provide the optimal transmission rate region of the two-receiver EBC in terms of the probability
	vector $\{\pv_s\}$. The formal proof is not given until a later section in which we prove the general $K$-receiver case. Then, we apply the result to derive the maximum
	data rate for the toy model, from which we can appreciate the performance gain brought
	by the location information. 
	\vspace{-0.1cm}
	\subsection{Main results}
	\begin{theorem} 
		\label{the:two_region}
		The normalized rate pair\footnote{Normalized by $\log|\Xc|$, i.e., measured as symbols per
			channel use.} $({R_1},{R_2})$ for the two-receiver EBC is achievable if and only if 
		for any $\mu \ge 1$, 
		\vspace{-0.1cm}
		\begin{align}
			{R_1}\! +\! \mu {R_2} &\le \onev^\T \max\bigl\{\pv_{\overbar{00}}, \mu
			\,\pv_{01,11}\bigr\},  \label{eq:outer_bound1} 
			\\
			{R_2}\! +\! \mu {R_1} &\le \onev^\T \max\bigl\{\pv_{\overbar{00}}, \mu \,\pv_{10,11}\bigr\},  \label{eq:outer_bound2}
		\end{align}%
		
		\vspace{-0.1cm}
		\noindent
		where the maximum between two vectors is component-wise. 
	\end{theorem}
	\begin{proof}
		The converse has been shown in \cite{li2017capacity}. The achievability can be shown in two steps. First we shall
		establish the achievability for the $K$-receiver case as an optimization problem. Then we let $K=2$ and establish
		the equivalence between the two-receiver achievable region and the above region. See the appendix for details. 
		\vspace{-0.1cm}
	\end{proof}

	In the case where the estimate is independent of the true channel state --- it is the case when
	the velocity of the vehicles is larger than $2R_B/Ts$ in the toy model --- the above theorem
	corresponds to the result in \cite{wang2012capacity,gatzianas2013multiuser} where only state
	feedback is exploitable. Indeed, in this case the expression can be simplified since
	each vector in the component-wise max contains identical elements. One can hence swap the inner
	product and the max operation and get the maximum between the marginal probabilities, namely, 
	$\max\bigl\{p_{\overbar{00}}, \mu \,p_{01,11}\bigr\}$ and 
	$\max\bigl\{p_{\overbar{00}}, \mu \,p_{10,11}\bigr\}$ on the right-hand side. Note that
	moving the inner product inside the maximum induces a loss in general when the state estimates
	are useful, i.e., when the components in $\pv_s$ are not identical. The subtle difference
	marks the potential gain that can be exploited using the estimated current state.

	Let us now consider the toy model with spatial independence and symmetric velocity. Then, it is
	not hard to verify that the region in Theorem \ref{the:two_region} is symmetric as well. The
	following proposition provides an explicit expression of the symmetric rate, the maximum rate that both
	vehicles can achieve simultaneously. 
	\vspace{-0.1cm}
	\begin{proposition}\label{pro:two_region_example}
		If the probability for the vehicles to be disconnected from the wireless radio sites is within a certain interval, say $1/3<{e^{-\lambda A}}<4/5$ and ${A_k} > {{\ln \left( {\frac{{{e^{\lambda A}} + 1}}{{3{e^{\lambda A}} - {e^{2\lambda A}}}}} \right)} \mathord{\left/
				{\vphantom {{\ln \left( {\frac{{{e^{\lambda A}} + 1}}{{3{e^{\lambda A}} - {e^{2\lambda A}}}}} \right)} \lambda }} \right.
				\kern-\nulldelimiterspace} \lambda }$, then 
		\vspace{-0.05cm}
		\begin{align}
			{R_{\text{sym}}^{\text{mixed}}}\! = \!\frac{{{e^{\lambda A}} \!-\! {e^{\lambda {A_k}}}}}{{2{e^{\lambda(A + {A_k})}}\! +\! {e^{\lambda (2A + {A_k})}}}} \!+ \!\frac{{(1 \!-\! {e^{ - \lambda A}})(1\! +\! {e^{ - \lambda A_k}})}}{{2\! + \!{e^{ - \lambda {A_k}}}}}.\nonumber
		\end{align}
		
		\vspace{-0.05cm}
		\noindent
		In particular, at high mobility, i.e., when $A_k = A$, we have
		\vspace{-0.1cm}
		\begin{align}
			R_{\text{sym}}^{\text{FB}} = \frac{{1 - {e^{ - 2\lambda A}}}}{{2 + {e^{ - \lambda A}}}},\nonumber
		\end{align}
		
		\vspace{-0.1cm}
		\noindent
		where ``FB'' means that only feedback is available.  
	\end{proposition}
	A sketch of proof is provided in the appendix. Remarkably, the second term in $R_{\text{sym}}^{\text{mixed}}$ is not less than
	$R_{\text{sym}}^{\text{FB}}$ in all cases. 
	Since the first term in $R_{\text{sym}}^{\text{mixed}}$ is always non-negative and is strictly
	positive when $v_kT_s<2R_B$, it can be regarded as the net performance gain from the location information.

	
	As a numerical example, Fig.~\ref{fig:region_60km} shows the rate regions for the toy model of different symmetric velocities, with  wireless radio site density
	$\lambda=4/\text{km}^2$, delay
	$T_s=10\text{\,s}$ and LoS range $R_B=0.2\text{\,km}$ \cite{bai2015coverage}. 
	When $v=0$, the current state information is perfectly known, while if $v\ge 2R_B/T_s = 144 \text{km/h}$, the state
	information is completely outdated. When the two vehicles are moving at a moderate speed of $60\text{\,km/h}$,
	the additional location information can provide a gain of $15.32\%$ on the symmetric rate. 
	Nevertheless, we see that even at high speed, the scheme still outperforms the orthogonal
	access, e.g., Time Division Multiple Access~(TDMA) strategy, using only outdated state information. From another perspective, the location information can help to reduce the infrastructure costs in low-medium velocity region. In Fig. \ref{fig:gain}, we plot the minimum  radio site density required to achieve a target symmetric rate ${R_{\text{sym}}^{\text{T}}}$ as a function of the velocity. For example, to achieve $\!{R_{\text{sym}}^{\text{T}}}=\!0.4$, when only feedback is available, it requires that $\lambda\approx\!9.8/\text{km}^2$, while $\lambda$ reduces to $8.2/\text{km}^2$ if both vehicles are moving at a velocity of $60\text{\,km/h}$. We can either deploy less  wireless radio sites or shut down some, which can cut down $16.6\%$ expenses.
	
	\begin{figure*}[!ht]
		\begin{minipage}{0.5\textwidth}
			\setlength{\abovecaptionskip}{0pt}
			\setlength{\belowcaptionskip}{0pt}
			\centering
			\includegraphics[height=6cm]{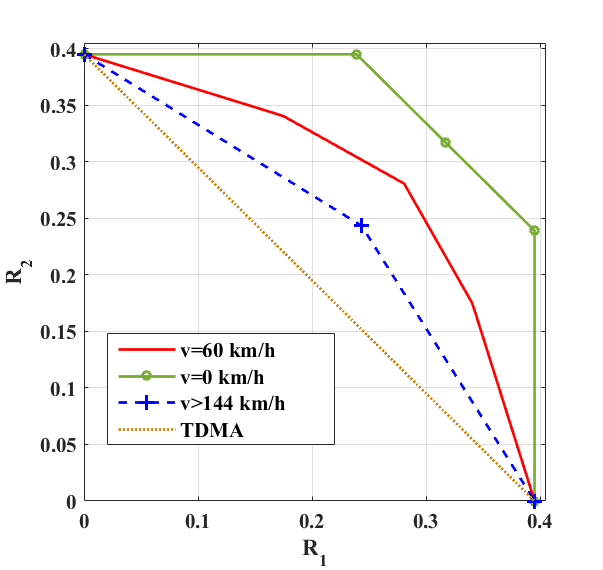}
			\caption{Optimal rate regions for the two-receiver V2I system,  with \protect\\$v_1=v_2=v$, $R_B=0.2\text{\,km}$, $T_s=10\text{\,s}$, $\lambda=4/\text{km}^2$.}\label{fig:region_60km}
		\end{minipage}
		\begin{minipage}{0.5\textwidth}
			\setlength{\abovecaptionskip}{0pt}
			\setlength{\belowcaptionskip}{0pt}
			\flushleft
			\includegraphics[width=8.6cm]{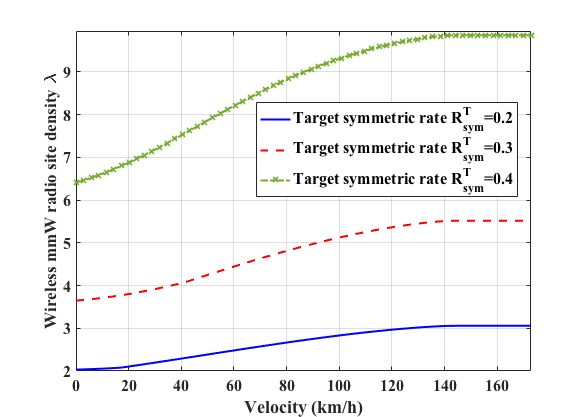}
			\caption{Minimum wireless radio site density required to achieve a target rate  ${R_{\text{sym}}^{\text{T}}}$ as a function of the symmetric velocity, with $R_B=0.2\text{\,km}$, $T_s=10\text{\,s}$.}\label{fig:gain}
		\end{minipage}
		\vspace{-0.15cm}
	\end{figure*}
	\vspace{-0.2cm}
	\subsection{(Informal) Scheme description}
	To illustrate the main idea behind the scheme, we first give the following high level description. The formal proof
	is deferred to the next section for the general $K$-receiver case. 
	
	The transmission of the proposed scheme consists of $B+1$
	blocks. In each block, the transmitter sends two types of information in two phases: private and common. 
	Let $V_{1,b}$ and $V_{2,b}$, $b=[1:B]$, be the private signals intended for receivers 1 and 2, respectively, at block
	$b$. Then the first phase is designed to send $V_{1,b}$ and $V_{2,b}$ in three different ways: $V_{1,b}$, $V_{2,b}$,
	and $V_{1,b}\otimes V_{2,b}$, where $\otimes$ simply means that it is some mixture between $V_1$ and $V_2$. The
	transmission of the first phase generates some side information that needs to be resent as the common information in
	the next block, $b+1$. In other words, the second phase of block $b$ is designed to send $V_{0,b}$ containing useful side information about the first phase in block $b-1$. The side information is generated only for a certain
	combination of transmitted signal and channel state, as shown in Table~\ref{tab:SI}.
	
	Note that the side information is not retransmitted directly, it needs to be compressed with distributed compression
	before the retransmission. Further, the side information is either useful to both receivers or useful to one
	receiver but is known to the other one (thus cost nothing to retransmit after distributed compression). More details are
	given in the next section. 

	Repeat the above steps for $B$ blocks, in block $B+1$, in order to
	recover the lost private signal of block $B$, the cloud will only send
	the common signal $V_{0,B}$. After the transmission, the receivers begin
	to decode individually using all the signals they have received.
	Receiver $k$ first decodes $V_{0,B}$ according to the received signal in
	block $B+1$, then uses the side information contained in $V_{0,B}$ to
	recover the private message $V_{k,B}$. This private message, together
	with the received signal in block $B$ can help decode $V_{0,B-1}$. This
	backward decoding process continues until the first block, at which
	point the whole message is recovered. 
	To summarize, in each block, there are two phases resulting four kinds of signal:
	$V_{0}$, $V_1$, $V_2$, and $V_1\otimes V_2$. To guarantee the successful decoding of the original message, the four
	signals need to be ``scheduled'' in an optimal way. As a matter of fact, the knowledge of the current state from the
	location information helps improve such scheduling and as a result enhances the transmission rate.  
	
	
	\section{The General $K$-receiver Scheme}\label{sec:proof}
	In this section, we describe the proposed scheme in detail and derive an achievable $K$-receiver EBC rate region. 
	
	First, we ignore the memory in the channel and assume that the process
	$\{S_t, \hat{S}_t\}_t$ is i.i.d.~over time, i.e., we only care about the
	correlation between $S_t$ and $\hat{S}_t$. Using a similar argument as
	given in \cite[Ch.7, Remark~7.4]{el2011network}, such an assumption does
	not induce rate loss\footnote{The bottom line is the communicating nodes
		can always choose to ignore temporal correlation.}. Then, the 1-to-$K$ EBC in Section \ref{sec:system_model} can be related to a
	special case of the general $K$-receiver discrete memoryless
	state-dependent BC $(\Xc\times\Sc\times\hat{\Sc},
	p(y|x,s)p(s,\hat{s}),\prod\nolimits_k \Yc_k)$ with probability mass
	function~(pmf)
	\vspace{-0.1cm}
	\begin{align}\label{eq:system_model}
		\textstyle\prod_{t = 1}^n { p( y_{1t},\ldots,y_{Kt}| x_t,{s_t}} )
		p(s_{t},\hat{s}_t),
	\end{align}
	
	\vspace{-0.1cm}
	\noindent
	where ${\sv} := \{s_t\}_t \in {\mathcal{S}^n}, \hat{\sv} :=\{\hat{s}_t\}_t \in
	{\hat{\Sc}^n}, \xv := \{x_t\}_t \in \Xc^n, \yv_{k}:=\{y_{k,t}\}_t \in\Yc_k^n$ are the
	sequences of the channel state, the estimated channel state, the channel input, and the channel
	output at receiver $k$ for $n$ time slots.  At time $t$, the past channel states,
	denoted by $\sv^{t-1}$, 
	is available to the transmitter perfectly, whereas the estimated current state $\hat{s}_t$
	can be obtained non-causally. At the end of the transmission of $n$ symbols, both $\sv$ and $\hat{\sv}$ are known to all the
	receivers for decoding. In particular, for the transmission of message $m_k$ to receiver $k\in\Kc$, with $m_k \in \Mc_k
	:= [1:2^{n R_k}]$, the encoding functions are $\{\phi_{t}: \Mc_1 \times \ldots\times \Mc_K \times \Sc^{t-1} \times \hat{\Sc}_{t} \rightarrow
	\Xc\}_{t=1}^n$, while the decoding function at receiver $k$ is $\varphi_{k}:\Yc_k^n \times \Sc^n \times \hat{\Sc}^n \rightarrow
	\Mc_k$. We say that the rate tuple $(R_1,\ldots,R_K)$ is achievable if
	the probability of error of each receiver goes to $0$ when $n\to\infty$.
	For the EBC we have $\Sc=\hat{\Sc}:=\{0,1\}^K$ and that the outputs are deterministic functions of the input given the state. Therefore,
	$\Yc_k=\Xc\cup\{?\}, k\in\Kc$. 
	
	\begin{theorem}\label{thm:general_region}
		A rate tuple $(R_1,\ldots,R_K)$ is achievable for the $K$-receiver EBC if 
		\begin{align}
			&R_k \leq I(V_k;Y_k,\{\hat{Y}_{{\Jc}}\}_{k\in\Jc,\Jc\subseteq \Kc}\cond S,Q,\hat{S}),\quad k\in\Kc\label{eq:general_region1}\\\vspace{-0.1cm}
			&\mathop {\max }\limits_{k,\Jc,k\in\Jc\hfill\atop\Jc\subseteq\Kc}\biggl\{ I(\hat{Y}_{{\Jc}};V_{{\Jc}},\{V_{k'}\}_{k'\in \Kc},\{V_{{\Ic} }\}_{\Ic \subset\Jc}\cond Y_k, S,Q,\hat{S})-\nonumber\\
			& I(V_{{\Jc}}; Y_k,  \{\hat{Y}_{{\Lc}}\}_{\Lc \supset\Jc}\cond V_k,\{V_{{\Ic}}\}_{\Ic \subset\Jc, k\in\Ic}, S,Q,\hat{S})\biggl\}\le0,\label{eq:general_region2}
		\end{align}
		for some pmf
		\begin{align}
			&p(x\cond\{v_k\}_{k\in \Kc},\{v_{{\Jc}}\}_{\Jc\subseteq\Kc},q) p(q\cond\hat{s})\textstyle\prod_{k = 1}^{{K}} {p(v_k\cond\hat{s})}\times\label{eq:pmf}\\\vspace{-0.1cm}
			&\textstyle\prod_{j=2}^{K}\textstyle\prod_{\Jc\subseteq \Kc,|\Jc|=j} p(v_{{\Jc}}\cond\hat{s})p(\hat{y}_{{\Jc} }|\{v_{{\Ic}}\}_{\Ic \subseteq\Jc},\{v_k\}_{k\in\Kc},q,s),\nonumber 
		\end{align}
	\end{theorem}
	\begin{table}[!t] 
		\centering
		\caption{Side information to be resent in the future.}
		\label{tab:SI}
		\begin{tabular}{r|c|c|c}
			& $V_{1}$  & $V_{2}$  & $V_{1} \otimes V_{2}$  \\\hline\hline
			$(S_{1},S_{2}) = (0,0)$ & $\emptyset$ & $\emptyset$ & $\emptyset$\\\hline
			$(S_{1},S_{2}) = (1,0)$ & $\emptyset$ & $V_2$ & $V_2$\\ \hline
			$(S_{1},S_{2}) = (0,1)$ & $V_1$ & $\emptyset$ & $V_1$\\ \hline
			$(S_{1},S_{2}) = (1,1)$ & $\emptyset$ & $\emptyset$ & $V_1$ or $V_2$\\
			\hline
		\end{tabular}
	\end{table}
	
	\noindent where ${\left\{{{V_\Ic}} \right\}_{\Ic \subset \Jc}}\!:=\!\left\{ {{V_\Ic}:\!\Ic\!\!\subset\!\!\Jc} \right\}$, ${\{\hat{Y}_{{\Lc}}\}_{\Lc \supset\Jc}}\!\!:=\!\left\{ {\hat{Y}_{{\Lc}}:\!\Lc \!\!\supset \!\!\Jc} \right\}$. $Q$ is the time-sharing Random Variable (RV), $V_k$ is the private signal for receiver $k$, $\hat{Y}_{\Jc}$ and $V_{\Jc}$ are respectively the side information intended for receivers in set $\Jc$ and the signal that carries such information. $\hat{Y}_{\Jc}$ encloses the related signals of the previous block, namely, all the private signals $\{V_k\}_{ k\in\Jc}$ and all the common signals $\{V_\Ic\}_{ \Ic\subset\Jc}$ intended for receiver sets which are strict subsets of $\Jc$. 
	
	Now, we describe formally the proposed scheme. Our scheme integrates the block-Markov scheme and the joint source-channel
	coding. In particular, the transmission consists of $B+K-1$ blocks each of length $n$. The message $m_k$ intended for receiver $k$, $k\in\Kc$, is divided into $B$ sub-messages,
	i.e., $m_{k,b}\in\Mc_k$, each one transmitted in block $b$, $b\in[1:B]$.
	\subsubsection*{\underline{Codebook generation}}
	Fix the pmf as described in \eqref{eq:pmf}.
	\begin{enumerate}
		\item Before each block $b$, randomly generate the time-sharing sequence $\qv_b$ according to $\prod\nolimits_{i= 1}^{{n}} {p(q_i\cond {{\hat s}_i})}$.
		\item At the beginning of each block, generate $2^{nR_k}$ sequences $\pmb{v}_k (m_k)$, $m_k \in
		[1:2^{nR_k}]$, randomly and independently for each receiver $k\in\Kc$, according to $\prod_{i =1}^{{n}} {p(v_{k,i}\cond {{\hat s}_i})}$.
		\item At the beginning of each block, for each $\Jc\subseteq\Kc$, randomly
		generate $2^{nR_{{\Jc}}}$ independent sequences $\pmb{v}_{{\Jc} } (m_{{\Jc}})$
		according to $\prod_{i=1}^n p(v_{{{\Jc}},i}\cond {{\hat s}_i})$ with $m_{{\Jc} } \in [1:2^{nR_{ {\Jc} }}]$.
		\item  At the end of each block, upon the reception of the state feedback $\sv$, randomly and
		independently generate $2^{nR_{{\Jc}}}$ sequences $\hat{\pmb{y}}_{{\Jc}} (m_{{{\Jc}}})$ for each $\Jc\subseteq \Kc$,
		according to $\prod_{i=1}^n p(\hat{y}_{{{\Jc} },i}|s_i,q_i)$.
	\end{enumerate}
	\subsubsection*{\underline{Encoding}}
	\begin{itemize}
		\item Set $m_{{\Jc},b} = 1$, $\forall \Jc \subseteq \Kc$, $b \in \{0,
		[B+|\Jc|-1:B+K-2]\}$. Set $m_{k,b}=1, \forall k\in\Kc$, $b \in
		[B+1:B+K-1]$.
		\item For each $\Jc\subseteq\Kc$, at the end of each block $b$, $b\in[1:B+|\Jc|-2]$, given the state feedback $\pmb{s}_{b}$, $\{m_{ {\Ic} ,b-1}\}_{\Ic \subseteq \Kc}$ and $\{m_{k,b}\}_{k \in \Kc}$, the encoder looks for a unique message index $m_{ {\Jc} ,b}$ such that $(\hfill\hat{\pmb{y}}_{{\Jc}}(m_{{\Jc},b}),\hfill \{\vv_{{\Ic}}(m_{{\Ic},b-1})\}_{\Ic \subseteq \Jc},\hfill\\\{\vv_k(m_{k,b})\}_{k\in \Kc}, \sv_b, \qv_b, \hat{\sv}_b)$ are jointly typical. If there is more than one index, it selects one of them uniformly at random, otherwise an error is
		declared. According to the covering lemma \cite{el2011network}, such an index can be found with high probability if
		\begin{align}\label{eq:RJ_LB}
			\!\!nR_{{\Jc}}\!\geq\! nI(\hat{Y}_{{\Jc}};\!\{V_k\}_{k \in \Kc},\! \{V_{{\Ic} }\}_{\Ic \subseteq \Jc}|S, Q, \hat{S})\!+\!n \epsilon_n.
		\end{align}
		
		\item In block $b\in[1:B+K-1]$, the transmitter generates a sequence $\pmb{x}$ from
		$(\{\pmb{v}_{{\Jc}}(m_{{\Jc},b-1})\}_{\Jc\subseteq\Kc},\{\pmb{v}_{k}(m_{k,b})\}_{k\in \Kc})$ according to $\prod_{i=1}^n
		p(x_i\cond \{v_{{\Jc},i}\}_{\Jc\subseteq\Kc}, \{v_{k,i}\}_{k\in \Kc}, q_i)$.
	\end{itemize}
	\subsubsection*{\underline{Decoding}}
	\begin{itemize}
		\item[-] At the end of the transmission, a backward decoding is performed at each receiver. For blocks $b\in[1:B+|\Jc|-2]$, for each $\Jc\subseteq \Kc, k\in\Jc$, with the
		knowledge of $\hat{m}_{k,b+1}$, $\{\hat{m}_{{\Ic} ,b}\}_{\Ic\subset\Jc}$, $\{\hat{m}_{{\Lc} ,b+1}\}_{\Lc \supset\Jc}$ and the state information, the receiver~$k$ finds a unique index
		$\hat{m}_{{\Jc},b}$, such that $(\hat{\pmb{y}}_{{\Jc}}(\hat{m}_{{\Jc},b}), \pmb{y}_{k,b}, \pmb{s}_{b}, \pmb{q}_{b},\hat{\pmb{s}}_{b})$  and $(\pmb{v}_{{\Jc}}(\hat{m}_{{\Jc},b}),\hfill\{\pmb{v}_{{\Ic} }(\hat{m}_{{\Ic},b})\}_{\Ic\subset \Jc, k\in\Ic},\hfill \{\hat{\pmb{y}}_{{\Lc} }(\hat{m}_{{\Lc},b+1})\}_{\Lc \supset\Jc},\\
		\pmb{v}_k(\hat{m}_{k,b+1}),
		\pmb{y}_{k,b+1},\pmb{s}_{b+1},\pmb{q}_{b+1},\hat{\pmb{s}}_{b+1})$ are simultaneously
		jointly typical. According to Tuncel's coding\cite{tuncel2006slepian,
			he2018KBC}, we have $\hat{m}_{{\Jc},b}=m_{{\Jc},b}$ with high probability if
		\begin{align}\label{eq:RJ_UB}
			&nR_{{\Jc}} \leq n I(\hat{Y}_{{\Jc}}; Y_k,|S,Q,\hat{S})+\\
			&nI(V_{ {\Jc}};Y_k, \{\hat{Y}_{{\Lc}}\}_{\Lc \supset\Jc}|V_k, \{V_{{\Ic}}\}_{\Ic\subset \Jc, k\in\Ic}, S, \hat{S}, Q)\!\!-\!n\epsilon_n'.\nonumber
		\end{align}
		
		\item[-] Given that $\{\hat{m}_{{\Jc},b}\}_{k\in\Jc, \Jc\subseteq \Kc}$ are available for $b\in[1:B]$, the receiver $k$ looks for a unique message index
		$\hat{m}_{k,b}$ such that $(\pmb{v}_k(\hat{m}_{k,b}), \hfill\pmb{y}_{k,b}, \hfill\{\hat{\pmb{y}}_{{\Jc} }(\hat{m}_{{\Jc},b})\}_{k\in\Jc,\Jc\subseteq \Kc},\hfill \pmb{s}_b,\hfill\pmb{q}_b,\\\hat{\pmb{s}}_b)$ are jointly typical. 
		We have $\hat{m}_{k,b}=m_{k,b}$ with high probability provided that
		\begin{align}\label{eq:Rk_UB}
			\!\!\!\!n R_k \leq n I(V_k;Y_k,\{\hat{Y}_{{\Jc} }\}_{k\in\Jc,\Jc\subseteq \Kc}\cond S,Q,\hat{S}) - n \epsilon_n''.
		\end{align}
	\end{itemize}
	
	From \eqref{eq:RJ_LB} to \eqref{eq:Rk_UB}, letting $n,B\to\infty$, and
	apply the Fourier-Motzkin elimination to all constraints, we obtain the
	rate region for the general case given in Theorem~\ref{thm:general_region}. To derive a tractable rate region, we apply  
	the following choices on the RVs. 
	\begin{itemize}
		\item The time-sharing RV $Q$ takes two kinds of values, namely, $Q \in
		\Qc:=\Qc_{\text{in}}\cup \Qc_{\text{mix}}$, where 
		$\Qc_{\text{mix}}:=\{i\otimes j,\;\forall i,j\in\Kc,i\ne j\}$ and $\Qc_{\text{in}}:=\{\Uc, \forall \;\Uc\subseteq\Kc,1\le|\Uc|\le K\}$. Further, $Q$ depends on the estimated state $\hat{S}$ as 
		\begin{align}\label{eq:def_Q_hats}
			\P(Q = q \cond \hat{S} = \hat{s}) = \alpha_{q,\hat{s}}, \quad q\in\Qc,
			\hat{s} \in {\Sc}, 
		\end{align}
		with $\alpha_{q,\hat{s}}\ge0$ and $\sum_{q\in\Qc} \alpha_{q,\hat{s}}=1$ for any $\hat{s}\in{\Sc}$.
		\item 
		When $Q = \Uc\in\Qc_{\text{in}}$, an individual (private or common) signal
		intended to the receiver group $\Uc$ is transmitted, i.e., $X = V_{\Uc}$.  
		The side information $\hat Y$ is a deterministic function of $(V,S,\Uc)$, that is
		\begin{itemize}
			\item if $S_{\Uc}=\onev$, then ${\hat{Y}_{\Jc}}=\emptyset,\;\;\forall \Jc\subseteq\Kc$;
			\item if $S_{\Uc}\neq\onev$, and there exits a set  $\Jc\supset \Uc$ such that $S_{\Jc\backslash \Uc}=\onev,S_{\Kc\backslash \Jc}=\zerov$ (defined as  $\Sc_c(\Uc,\Jc)$), then
			\begin{equation}
			{\hat{Y}_{\tilde{\Jc}}}=X,  \text{w.p. } \beta_{\Uc,\tilde{\Jc}}^S,\quad\quad\forall  \tilde{\Jc}\in\Jc_c(\Uc,\Jc),
			\end{equation}
			where $\beta_{\Uc,\tilde{\Jc}}^S\ge0$, $\Jc_c(\Uc,{\Jc})$ is the set of $\tilde{\Jc}$ such that $\Uc\subset\tilde{\Jc} \subset \Jc, |\tilde{\Jc}|=|\Jc|-1 \text{ or } \tilde{\Jc}=\Jc$.
		\end{itemize}
		\item When $Q = (i\otimes j)\in\Qc_{\text{mix}}$, 
		we send a mixture of private signals, $X = V_i\otimes V_j$, the side information is set as
		\begin{equation*}\label{eq:def_side_info2}
			{\hat{Y}_{\{i,j\}}} = \begin{cases}
				V_i, &\text{if } S_{\Kc\backslash\{i,j\}}=\zerov,S_j=1,\\
				V_j, &\text{if } S_{\Kc\backslash\{i,j\}}=\zerov,S_{\{i,j\}}=10,\\
				0, & \text{otherwise.}
			\end{cases}
		\end{equation*}
		\item Uniformly distributed $V$'s in $\mathcal{X}$, namely, $\forall x\in\Xc$,
		\begin{align}
			\!\!\!\!\P[V_i=x] =\P[V_{{\Jc}}=x]={1 \mathord{\left/
					{\vphantom {1 |\mathcal{X}|}} \right.
					\kern-\nulldelimiterspace} |\mathcal{X}|},\,i\in\Kc,\Jc\subseteq\Kc.
		\end{align}
	\end{itemize}
	With the above setting, we obtain the result as follows.
	\begin{proposition}\label{the:K_region}
		A normalized rate tuple $({R_1},\ldots,{R_K})$ for the $K$-receiver EBC is achievable if 
		\begin{align}
			\!\!\!R_k\!\le\!\biggl({{\!\alphav _{k}^\T}{{ 
						\,\pv}_{{{s_{\Kc}\ne\zerov}}}}}\!+\!\!\!\!\!\!\!
			\sum\limits_{j\ne k,
				{j\in\Kc}}{{\!\!\!\!\!\alphav _{k \otimes
						j}^\T}\,{\pv_{{s_{\Kc\backslash \{ k,j\} }}
						=\zerov,{s_{\{k,j\}}} \ne\zerov}}}\!\biggr), 
		\end{align}
		$\forall k\in\Kc$, with the following constraints:
		\begin{align}
			&\mathop {\max}_{{k,\Jc:\hfill\atop k\in\Jc}\atop{\Jc\subseteq\Kc}}\biggl\{\sum\limits_{{ \Lc\supseteq\Jc}\atop{|\Lc|-|\Jc|=0,1}}\!\sum\limits_{s\in\Sc_c(k,\Lc)}\alphav_k^\T\pv_s\beta_{k,\Jc}^s+\nonumber\\
			&\sum\limits_{\Ic\subseteq\Jc\atop k\in\Ic}\sum\limits_{{ \Lc\supseteq\Jc}\atop{|\Lc|-|\Jc|=0,1}}\!\!\sum\limits_{s\in\Sc_c(\Ic,\Lc)}\alphav_{\Ic}^\T\pv_s\beta_{\Ic,\Jc}^s-\alphav_\Jc^\T \pv_{s_k=1}\nonumber\\
			&-\sum\limits_{\Lc\supset\Jc}\sum\limits_{{ \Uc\supseteq\Lc}\atop{|\Uc|-|\Lc|=0,1}}\sum\limits_{s\in\Sc_c(\Jc,\Uc)}\alphav_{\Jc}^\T\pv_s\beta_{\Jc,\Lc}^s\nonumber\\
			&+\sum\limits_{j \ne k,j\in\Kc}
			\!\!\!\!{{\alphav _{k \otimes
						j}^\T}\,{\pv_{{s_{\Kc\backslash \{ k,j\} }} =\zerov,{s_{\{k,j\}}}
						\ne\zerov}}}\mathbbm{1}_{\{\Jc=\{k,j\}\}}\biggr\}\le 0,	\end{align}
		and $\forall\;\Uc\subset \Jc, s\in\Sc_c(\Uc,\Jc)\;$, 
		\begin{align}
			\begin{cases}
				\beta_{\Uc,{\Jc}}^s\in [0,1],&\text{if } \Jc=\Kc, |\Uc|=|\Jc|-1,\\
				\sum\limits_{\tilde{\Jc}\in\Jc_c(\Uc,{\Jc})}\!\!\!\!\!\!\!\beta_{\Uc,\tilde{\Jc}}^s=1, &\text{otherwise,}\label{eq:constraint_beta}
			\end{cases}
		\end{align}
		
		\vspace{-0.1cm}
		\noindent where $\alphav_{q}:=[\alpha_{q,\hat{s}} : \hat{s}\in{\mathcal{S}}]\in[0,1]^{2^K}, q\in\Qc$ with $\sum\limits_{q\in \Qc} {\alphav_{q}}=\onev\label{eq:constraint_alpha}$.
	\end{proposition}
	
	The intuition behind the setting of the side information is the
	following. When we transmit an individual signal, i.e., $Q= \Uc$, the
	signal intended for receiver set $\Uc$ is sent. If this signal is not received by
	some of them (${S_{{\Uc}}} \ne \bm{1}$), and meanwhile is overheard
	by some unintended receivers (denoted by $\Jc\backslash {\Uc}$), then this
	signal becomes a side information for the receiver group $\Jc$ and will be
	compressed in $\hat{Y}_{\Jc}$ and transmitted in the future. In 
	\cite{he2018KBC}, the set $\Jc$ depends on the channel state in a
	deterministic way such that ${S_{{\Kc\backslash\Jc}}} =\bm{0}$. 
	For example, in the three-receiver case, assume
	$\Uc=\{1\}$, then ${\hat{Y}_{{\{1,2,3\}}}}=V_1$ as long as $S=011$. But such a setting may be suboptimal when
	the channel is asymmetric. In fact, even when $S=011$, it may be more efficient to
	retransmit $V_1$ in ${\hat{Y}_{{\{1,2\}}}}$ than in ${\hat{Y}_{{\{1,2,3\}}}}$
	depending on the channel statistics. One crucial idea in our current work is to enable
	such a \emph{downgrading} by introducing the conditional probabilities $\beta_{\Uc,\Jc}^S$ and $\beta_{\Uc,\tilde{\Jc}}^S$. Specifically, given $\Uc$ and the state, these $\beta$'s control how much side information can be compressed in ${\hat{Y}_{\Jc}}$ and in its lower layer ${\hat{Y}_{\tilde{\Jc}}}$ respectively.

	\section{Conclusion}\label{sec:conclusion}
	Through an information-theoretic study, we have demonstrated the potential benefit of vehicle location information to downlink rate improvement in a mmW V2I network. An interesting future direction is to investigate practical coding schemes for such networks.

\section{Appendix}
\subsection{Achievability of Theorem \ref{the:two_region}}
By setting $K=2$ in Proposition~\ref{the:K_region}, we obtain the normalized achievable rate region of two-receiver EBC as follows.
\begin{proposition}\label{pro_twoEBC}
	Let us define
	\begin{align}
		\Rc&_{\text{2EBC}}(\alphav_1, \alphav_2,\alphav_{1\otimes2}):=\nonumber\\
		&\left\{(R_1, R_2):
		\begin{array}{l}
			R_1  \leq (\alphav_1+\alphav_{1\otimes2})^\T \pv_{\overbar{00}}\\
			R_2  \leq (\alphav_2+\alphav_{1\otimes2})^\T \pv_{\overbar{00}}  
		\end{array}
		\right\}.
		\label{eq:def_achivable_region1}
	\end{align}%
	Then the achievable region of the two-receiver EBC is the convex hull of the union of
	$\Rc_{\text{2EBC}}(\alphav_1, \alphav_2,\alphav_{1\otimes2})$ over all
	$\alphav_1, \alphav_2, \alphav_{1\otimes2}\in [0,1]^4$ such that
	$\alphav_1+\alphav_2+\alphav_{1\otimes2}\preceq \onev$ and
	\begin{align}
		(\alphav_2+\alphav_{1\otimes2})^\T\pv_{10,11} + (\alphav_1+\alphav_{1\otimes2})^\T\pv_{\overbar{00}} &\leq p_{10,11},
		\label{eq:def_achivable_region20}
		\\
		(\alphav_1+\alphav_{1\otimes2})^\T\pv_{01,11} + (\alphav_2+\alphav_{1\otimes2})^\T\pv_{\overbar{00}} &\leq p_{01,11}.
		\label{eq:def_achivable_region2}
	\end{align}%
\end{proposition}
Now we need to establish the equivalence between the above achievable region and the region in Theorem \ref{the:two_region}. In the two-receiver case, the outer bound in Theorem \ref{the:two_region} can be reformed as \cite{li2017capacity}: 
	\begin{align}
		&\Rc_{\text{2EBC}}^o(\betav_1, \betav_2):=\nonumber\\
		&\left\{
		\begin{array}{l}
			R_1\le \text{min}\left\{\betav_1^\T \pv_{\overbar{00}},(\onev- \betav_2)^\T\pv_{10,11}\right\}\\
			R_2\le \text{min}\left\{\betav_2^\T \pv_{\overbar{00}},(\onev- \betav_1)^\T\pv_{01,11}\right\}
		\end{array}\right\}.\label{eq:outer_2EBC}
	\end{align}%
	for $\zerov\preceq\betav_1,\betav_2\preceq \onev$. In the outer bound \eqref{eq:outer_2EBC}, both $R_1$ and $R_2$ take minimum value in two possible candidates, let us consider the following 2 cases out of the 4 possible combinations.
	
	First, when
	$\left\{
	\begin{array}{l}
	\betav_1^\T \pv_{\overbar{00}}>(\onev- \betav_2)^\T\pv_{10,11}\\
	\betav_2^\T \pv_{\overbar{00}}\le(\onev-
	\betav_1)^\T\pv_{01,11}\end{array}\right.$, the outer bound
	is written as
	\begin{align}
		\Rc_{\text{2EBC}}^o(\betav_1, \betav_2)= \left\{
		\begin{array}{l}
			R_1  \le (\onev- \betav_2)^\T\pv_{10,11}\\
			R_2  \le \betav_2^\T \pv_{\overbar{00}}\label{eq:theorem3_case1}
		\end{array}\right\}.
	\end{align}%
	Since $\betav_1^\T \pv_{\overbar{00}}>(\onev- \betav_2)^\T\pv_{10,11}$ and $\zerov^\T \pv_{\overbar{00}}\le(\onev- \betav_2)^\T\pv_{10,11}$, there exists a $\zerov\prec\bm{\eta}\preceq \betav_1$ such that $\betav_1^*=\betav_1-\bm{\eta}$ and $\left\{\begin{array}{l}
	\betav_1^{*\T} \pv_{\overbar{00}}=(\onev- \betav_2)^\T\pv_{10,11}\\
	\betav_2^\T \pv_{\overbar{00}}<(\onev- \betav_1^*)^\T\pv_{01,11}
	\end{array}\right.$. With
	$\betav_1^*$ and $\betav_2$, we have
	\begin{align}
		\Rc_{\text{2EBC}}^o(\betav_1^*, \betav_2)=\left\{
		\begin{array}{l}
			R_1\le \betav_1^{*\T} \pv_{\overbar{00}}\\
			R_2\le \betav_2^\T \pv_{\overbar{00}}\label{eq:theorem3_case1_R*}
		\end{array}\right\},
	\end{align}%
	which is equivalent to the original outer bound. 
	
	Next, when 
	$\left\{
	\begin{array}{l}
	\betav_1^\T \pv_{\overbar{00}}>(\onev- \betav_2)^\T\pv_{10,11}\\
	\betav_2^\T \pv_{\overbar{00}}>(\onev- \betav_1)^\T\pv_{01,11}
	\end{array}\right.
	$, we have
	\begin{align}
		\Rc_{\text{2EBC}}^o(\betav_1, \betav_2)=\left\{
		\begin{array}{l}
			R_1  \le  (\onev- \betav_2)^\T\pv_{10,11}\\
			R_2  \le (\onev- \betav_1)^\T\pv_{01,11}
		\end{array}\right\}.\label{eq:theorem3_case3}
	\end{align}%
	Note that there exists a $\zerov\prec\bm{\eta_1}\preceq \betav_1$ such that $\betav_1^{*\T} \pv_{\overbar{00}}=(\onev- \betav_2)^\T\pv_{10,11},\label{eq:theorem3_case3_beta1}$ where $\betav_1^*=\betav_1-\bm{\eta_1}$.
	If $\betav_2^\T \pv_{\overbar{00}}\le(\onev- \betav_1^*)^\T\pv_{01,11}$, then the new outer bound can be written as
	\begin{align}\label{eq:theorem3_case3_R*1}
		\Rc_{\text{2EBC}}^o(\betav_1^*, \betav_2)=\left\{
		\begin{array}{l}
			R_1  \le \betav_1^{*\T} \pv_{\overbar{00}}\\
			R_2  \le\betav_2^\T \pv_{\overbar{00}}
		\end{array}\right\},
	\end{align}%
	which contains the original outer bound $\Rc_{\text{2EBC}}^o(\betav_1, \betav_2)$.
	If $\betav_2^\T \pv_{\overbar{00}}>(\onev- \betav_1^*)^\T\pv_{01,11}$, then there exists a $\zerov\prec\bm{\eta_2}\preceq \betav_2$ such that $\betav_2^{*\T} \pv_{\overbar{00}}=(\onev- \betav_1^*)^\T\pv_{01,11}$, where $\betav_2^*=\betav_2-\bm{\eta_2}$.
	Now one can verify that $\left\{
	\begin{array}{l}
	\betav_1^{*\T} \pv_{\overbar{00}}<(\onev- \betav_2^*)^\T\pv_{10,11}\\
	\betav_2^{*\T} \pv_{\overbar{00}}=(\onev- \betav_1^*)^\T\pv_{01,11}
	\end{array}\right.$, as such, the new outer bound becomes
	\begin{align}
		\Rc_{\text{2EBC}}^o(\betav_1^*, \betav_2^*)=\left\{
		\begin{array}{l}
			R_1  \le \betav_1^{*\T} \pv_{\overbar{00}}	\\
			R_2  \le\betav_2^{*\T} \pv_{\overbar{00}}
		\end{array}\right\},\label{theorem3_case_R*2}
	\end{align}%
	which contains the original outer bound.         
	
	As discussed above, for all the $\betav_1,\betav_2$ that do not simultaneously satisfy $\left\{
	\begin{array}{l}
	\betav_1^\T \pv_{\overbar{00}}\le(\onev- \betav_2)^\T\pv_{10,11}\\
	\betav_2^\T \pv_{\overbar{00}}\le(\onev- \betav_1)^\T\pv_{01,11}
	\end{array}\right.$, one can construct a new pair $(\betav_1^*,\betav_2^*)$ such that the two
	constraints are both satisfied while the new outer bound
	$\Rc_{\text{2EBC}}^o(\betav_1^*,\betav_2^*)$ contains the original outer bound region. In
	all cases, the larger outer bound region $\Rc_{\text{2EBC}}^o(\betav_1^*,\betav_2^*)$ is contained in the achievable region
	$\Rc_{\text{2EBC}}(\alphav_1^*,\alphav_2^*,\alphav_{1\otimes2}^*)$ by letting
	$\alphav_{1\otimes2}^*=\min\{\betav_1^*,\betav_2^*\},
	\alphav_1^*=\max\{\betav_1^*-\betav_2^*,\zerov
	\},
	\alphav_2^*=\max\{\betav_2^*-\betav_1^*,\zerov\}$, which completes the proof of equivalence between two-receiver achievable rate region and the outer bound region in Theorem \ref{the:two_region}, and therefore demonstrates the achievability of Theorem \ref{the:two_region}.

\subsection{Sketch of Proof of Proposition \ref{pro:two_region_example}}
Since the two vehicles are spatially independent and symmetric, we first write the symmetric rate as
\begin{align}
	{R_{\text{sym}}^{\text{mixed}}}=\mathop {\min}\limits_{\mu  \ge 1} \frac{\onev^\T \max\bigl\{\pv_{\overbar{00}}, \mu
		\,\pv_{01,11}\bigr\}}{{1 + \mu }}.\label{eq:def_Rsym}
\end{align}
It is not hard to justify that the optimal $\mu$ is among the following four possible candidates $[\mu_1,\mu_2,\mu_3,\mu_4]^T=\pv_{\overbar{00}}./\pv_{01,11}$, $\pv_{\overbar{00}}, \pv_{01,11}$ are derived according to equations \eqref{eq:prob_vec} and \eqref{eq:def_S}-\eqref{eq:def_cond_S2}.  A fixed ordering $\mu_2\le\mu_4\le\mu_1\le\mu_3$ can be expected, and therefore,
\begin{align}
{R_{\text{sym}}^{\text{mixed}}}=\min\left\{{R_{\text{sym},1},R_{\text{sym},2},R_{\text{sym},3},R_{\text{sym},4}}\right\},
\end{align}
where 
	\begin{align}
		&{R_{\text{sym},1}}= \frac{{1 - {e^{ - 2\lambda A}}}}{{\frac{{\left( {{e^{ - \lambda A}} - 1} \right)\left( {\frac{{{e^{ - \lambda \left( {A + {A_k}} \right)}}\left( {{e^{ - \lambda {A_k}}} - 1} \right)}}{{{e^{ - \lambda A}} - 1}} - 1} \right)}}{{{e^{ - \lambda \left( {A + {A_k}} \right)}} - 2{e^{ - \lambda A}} + 1}} + 1}},\label{eq:Rsym1}\\[8pt]
		&{R_{\text{sym},2}} = - \frac{{{\sigma _2}\sigma _3^2 + {e^{ - 2\lambda A}}\left( {{e^{ - 2\lambda {A_k}}} - 1} \right) - {e^{ - \lambda A}}{\sigma _2}{\sigma _3}}}{{\frac{{{\sigma _2}{\sigma _3}}}{{{\sigma _1} - 2{e^{ - \lambda A}} + 1}} + 1}}\nonumber\\
		&\qquad\qquad  + \frac{{{e^{ - \lambda A}}\left( {{\sigma _1}{\sigma _4} - {\sigma _3}} \right)}}{{\frac{{{\sigma _2}{\sigma _3}}}{{{\sigma _1} - 2{e^{ - \lambda A}} + 1}} + 1}},\label{eq:Rsym2}\\[8pt]
		&{R_{\text{sym},3}} =\! = \!\frac{{{e^{\lambda A}} \!-\! {e^{\lambda {A_k}}}}}{{2{e^{\lambda(A + {A_k})}}\! +\! {e^{\lambda (2A + {A_k})}}}} \!+ \!\frac{{(1 \!-\! {e^{ - \lambda A}})(1\! +\! {e^{ - \lambda A_k}})}}{{2\! + \!{e^{ - \lambda {A_k}}}}},\label{eq:Rsym3}\\[8pt]
	&{R_{\text{sym},4}} = \frac{{\left( {1 - {e^{ - \lambda A}}} \right){\sigma _5}}}{{\left( {{\sigma _5} + {e^{ - \lambda {A_k}}} - 1} \right)}},\label{eq:Rsym4}\\[8pt]
	&{\sigma _1} = {e^{ - \lambda \left( {A + {A_k}} \right)}},\\
	&{\sigma _2} = \frac{{{e^{ - 2\lambda A}}\sigma _4^2}}{{\sigma _3^2}} - 1,\\
	&{\sigma _3} = {e^{ - \lambda A}} - 1,\\
		&{\sigma _4} = {e^{ - \lambda A_k}} - 1,\\
	&{\sigma _5} = \frac{{{e^{ - \lambda \left( {A + {A_k}} \right)}}\left( {{e^{ - \lambda {A_k}}} - 1} \right)}}{{{e^{ - \lambda A}} - 1}} - 1.
	\end{align}

However, it is not easy to compare $R_{\text{sym},i}, i\in\left\{1,2,3,4\right\}$ in general cases, therefore, in Proposition~\ref{pro:two_region_example}, we give the conditions $1/3<{e^{-\lambda A}}<4/5$ and ${A_k} > {{\ln \left( {\frac{{{e^{\lambda A}} + 1}}{{3{e^{\lambda A}} - {e^{2\lambda A}}}}} \right)} \mathord{\left/
		{\vphantom {{\ln \left( {\frac{{{e^{\lambda A}} + 1}}{{3{e^{\lambda A}} - {e^{2\lambda A}}}}} \right)} \lambda }} \right.
		\kern-\nulldelimiterspace} \lambda }$, under which $R_{\text{sym},3}$ can be verified to be the smallest one. To emphasize,  ${R_{\text{sym}}^{\text{mixed}}}$ is not less than $R_{\text{sym}}^{\text{FB}}$ in all cases, here we just show a simple form of ${R_{\text{sym}}^{\text{mixed}}}$ under the specific conditions, so that one can compare ${R_{\text{sym}}^{\text{mixed}}}$ and $R_{\text{sym}}^{\text{FB}}$ explicitly and therefore has an intuitive understanding of the net performance gain obtained from the location information. 
\bibliographystyle{IEEEtran}
\bibliography{IEEEabrv,GC_2018}

\end{document}